\documentclass[paper=letter]{scrartcl}
\usepackage{amssymb,amsmath,amsthm,sectsty,url,graphicx}

\usepackage{microtype}

\usepackage{tikz} 		% used by comment macros		
\usetikzlibrary{positioning}
\usetikzlibrary{calc}
\usetikzlibrary{arrows}
\usetikzlibrary{decorations.markings}
\usetikzlibrary{decorations.pathmorphing}

\newsavebox{\theorembox}
\newsavebox{\lemmabox}
\newsavebox{\corollarybox}
\newsavebox{\propositionbox}
\newsavebox{\examplebox}
\newsavebox{\conjecturebox}
\newsavebox{\algbox}
\newsavebox{\qbox}
\newsavebox{\problembox}
\newsavebox{\definitionbox}
\newsavebox{\assumptionbox}
\newsavebox{\hypothesisbox}
\savebox{\theorembox}{\noindent\bf Theorem}
\savebox{\lemmabox}{\noindent\bf Lemma}
\savebox{\corollarybox}{\noindent\bf Corollary}
\savebox{\propositionbox}{\noindent\bf Proposition}
\savebox{\examplebox}{\noindent\bf Example}
\savebox{\conjecturebox}{\noindent\bf Conjecture}
\savebox{\algbox}{\noindent\bf Algorithm}
\savebox{\qbox}{\noindent\bf Question}
\savebox{\definitionbox}{\noindent\bf Definition}
\savebox{\problembox}{\noindent\bf Problem}
\savebox{\assumptionbox}{\noindent\bf Assumption}
\savebox{\hypothesisbox}{\noindent\bf Hypothesis}
\newtheorem{theorem}{\usebox{\theorembox}}
\newtheorem{lemma}[theorem]{\usebox{\lemmabox}}

\newtheorem{definition}{\usebox{\definitionbox}}

\theoremstyle{definition}
\newtheorem{example}[theorem]{Example}

\begin{document}

\title{Contagious Sets in Dense Graphs\footnote{A preliminary version of this paper was presented at the International Workshop on Combinatorial Algorithms (IWOCA) in 2015. The proceedings will be available at Springer.}}

\author{Daniel Freund\thanks{Center for Applied Mathematics, Cornell University, Ithaca, NY, USA.  Email: \texttt{df365@cornell.edu}. Supported in part by  U.S.\ Army Research Office grant W911NF-14-1-0477.}
\and
Matthias Poloczek\thanks{School of Operations Research and Information Engineering, Cornell University, Ithaca, NY, USA. Email:~\texttt{poloczek@cornell.edu}. Supported by the Alexander von Humboldt Foundation within the Feodor Lynen program, and in part by NSF grant CCF-1115256, and AFOSR grants FA9550-15-1-0038 and FA9550-12-1-0200.}
\and
Daniel Reichman\thanks{Department of Computer Science, Cornell University, Ithaca, NY, USA.  Email: \texttt{daniel.reichman@gmail.com}. Supported in part by NSF grants IIS-0911036 and CCF-1214844, AFOSR grant FA9550-08-1-0266, and ARO grant W911NF-14-1-0017.}}

\date{}

\maketitle

\begin{abstract}
We study the activation process in undirected graphs known as bootstrap percolation:
a vertex is active either if it belongs to a set of initially activated
vertices or if at some point it had at least~$r$ active neighbors, for a threshold~$r$ that is identical for all vertices.
A contagious set is a vertex set whose activation results with the entire
graph being active.
Let~$m(G,r)$ be the size of a smallest contagious set in a graph~$G$ on~$n$ vertices.

We examine density conditions that ensure~$m(G,r) = r$ for all~$r \geq 2$. 
With respect to the minimum degree, we prove that such a smallest possible contagious set is guaranteed to exist if and only if~$G$ has minimum degree at least~$\frac{k-1}{k}\cdot n$. 
Moreover, we study the speed with which the activation spreads and provide tight upper bounds on the number of rounds it takes until all nodes are activated in such a graph.

We also investigate what average degree asserts the existence of small contagious sets.
For $n \geq k \geq r$, we denote by $M(n,k,r)$ the maximum number of edges in an $n$-vertex graph $G$ satisfying $m(G,r)>k$.
We determine the precise value of $M(n,k,2)$ and $M(n,k,k)$, assuming that~$n$ is sufficiently large compared to~$k$. 
\end{abstract}
\newpage
\section{Introduction}
In this article we study the $r$-\emph{neighbor bootstrap percolation} process.
Here we are given an undirected graph $G=(V,E)$ and an integer~$r \geq 1$.
%
% % % MP: Should we change the lb on the threshold from 1 to 2? In extremal we talk about k=1, but well...
%
Every vertex is either \emph{active} or \emph{inactive}. We say a set $A$ of vertices is active if all vertices in $A$ are active. The vertices that are active initially are called \emph{seeds}, and the set of seeds is denoted by~$A_0$. If vertices become active thereafter we also refer to them as \emph{infected}.
A contagious process evolves in discrete rounds. The set of active vertices in round~$i>0$ is
$$A_i=A_{i-1}\cup \{v:|N(v)\cap A_{i-1}|\geq r\},$$
where $N(v)$ is the set of neighbors of $v$. That is, a vertex becomes active irrevocably in a given round if it has at least~$r$ active neighbors. We refer to~$r$ as the \emph{threshold}.
Let~$\langle A_0 \rangle$ be the set of nodes that will eventually become infected if we activate~$A_0$.
\begin{definition}
	Given $G=(V,E)$, a set $A_0\subseteq V$ is called \emph{contagious} if $\langle A_0 \rangle=V$. In words, activating $A_0$ results in the infection of the entire vertex set.
	The size of the smallest contagious set is denoted by $m(G,r)$. For a contagious set $A_0$, the number of rounds until total infection is the smallest integer~$t$ with $A_t = V$.
\end{definition}

The term bootstrap percolation is used sometimes to model the case where the seeds are chosen independently at random. In this work we use this term also with respect to the deterministic selection of a contagious set.
Bootstrap percolation was first studied by statistical physicists \cite{Chal}.
Since then, this model has found applications in many fields. For example, this model is related to ``word of mouth'' effects occurring in viral marketing, where the information is only revealed to a small group of persons initially, who subsequently share it with their friends resulting with a cascade that may spread to the entire network. 
Similarly, we can think of cascading effects in finance, where an institute might default if a certain number of business partners fail (cp.\ \cite{KempeKT15,amini2013resilience,NNUW13,chen09} and the references therein).
Furthermore, various questions related to bootstrap percolation have been examined for a large variety of graphs including hypercubes \cite{BB}, grids \cite{Lattice,square}, several models of random graphs \cite{JanLuc,AF,BP}, and expanders \cite{Expanders}.

A natural question is to determine for a given integer~$k$, what combinatorial properties of graphs ensure that the minimum size of a contagious set is at most~$k$. Such a characterization seems difficult even for $k=2$ (and $r=2$). Indeed the family of all graphs with a contagious set of size two include, for example, cliques, bipartite cliques (with both sides larger than one), and binomial random graphs with edge probability $p \geq n^{-1/2+\epsilon}$ \cite{JanLuc}.

Previous works have examined the connection between~$m(G,r)$ and the degree sequence of~$G$ \cite{Ackerman,Reichman}.
Here we continue this line of investigation and study two basic (and interrelated) graph parameters: the minimum degree and edge cardinality.
More concretely, our goal is to determine what conditions on these parameters imply that $m(G,r)=k$ where $k$ is small compared to the number of vertices in $G$, and $r \leq k$. We study the cases $r=k$ and~$r=2$.

How large does the minimum degree have to be in order to guarantee a contagious set of size~$k \geq 2$, if all thresholds are~$k$?
%
% % % MP: Adjusted because of my new proof:
%
Clearly, such a contagious set has minimum cardinality if it exists.

We prove that~$\left\lceil\frac{k-1}{k} \cdot n \right\rceil$ suffices, where~$n$ is the number of vertices. In particular, if~$k=2$ then the required minimum degree is~$\left\lceil \frac{n}{2} \right\rceil$.
A graph with this property is called Dirac graph.
We also show that this condition on the minimum degree is the best possible.
For~$k=2$ this is easy to see: if we lower the minimum degree to~$\left\lceil \frac{n}{2}-1 \right\rceil$, then~$G$ may be disconnected implying that $m(G,2)>2$ (provided that~$G$ has at least three vertices).
Graphs with minimum degree~$n/2$ are called Dirac graphs. In the Appendix we demonstrate that a contagious set of size 2 also exists in a generalization of Dirac graphs known as Ore graphs 
%
%We also prove the existence of a contagious set of size two for Ore graphs. Ore graphs are a generalization of Dirac graphs, where each pair of nonadjacent vertices $u,v$ obeys $\mathrm{deg}(u)+\mathrm{deg}(v)\geq n$.
%
(Dirac graphs and Ore graphs are known to have a Hamiltonian cycle).

% % % MP: New intro to speed of spreading
While the minimum size of contagious sets has been studied thoroughly, much less is known on the number of rounds that the activation process takes to infect the whole graph (see~\cite{time,Bolo1,Bolo2} for instance).
%
%For Dirac graphs which are not isomorphic to two cliques of equal size connected by a perfect matching, we are able to derive upper bounds on the number of rounds required to infect the whole of $G$. More specifically, we show that for such graphs, \emph{all} subsets of three nodes are contagious, and that any such subset will infect the whole graph in at most three rounds. Observe that it is easy to determine the family of all contagious sets in the Dirac graph consisting of two cliques connected by a perfect matching (as well as the number of rounds until all nodes are infected).
%
We study the speed of spreading (also referred to as maximal percolation time) and show for graphs with minimum degree at least~$\frac{k-1}{k} \cdot n$ that any contagious set activates the whole graph within four rounds if~$k=2$. In the case of~$k\geq 3$ even three rounds suffice. We provide instances that demonstrate that these upper bounds are tight.

\medskip

A classic question in graph theory is to determine the minimum number of edges in an $n$-vertex graph $G$ that ensures that $G$ possesses a monotone graph property. This article examines extremal questions related to the existence of small contagious sets.
\begin{definition}
	Given integers $n \geq k \geq r$, we denote by $M(n,k,r)$ the maximum number of edges in an~$n$ vertex graph $G$, where~$G$ satisfies~$m(G,r)>k$.
\end{definition}
First we study~$M(n,k,k)$: a necessary condition for a graph $G=(V,E)$ of $n>k$ vertices to satisfy $m(G,k)=k$ is that $G$ is \emph{connected}. Here we show that the minimum number of edges that guarantees connectivity is also sufficient to ensure~$m(G,k)=k$ for~$n \geq 2k+2$, i.e.\ we have~$M(n,k,k) = \binom{n-1}{2}$.
Next we consider the case where~$r=2$. For $k \ll n$, we prove that $M(n,k,2)=\binom{n-k+1}{2}+\lfloor \frac{k+1}{2}\rfloor-1$ holds.

%\newpage

\subsubsection*{Preliminaries}
All graphs are undirected. Given a graph $G=(V,E)$, we will always assume it has~$n$ vertices. The degree of a node $v \in V$ is denoted by $\mathrm{deg}(v)$. The set of all neighbors of a vertex $v$ are denoted by $N(v)$.
For a set~$S \subseteq V$ we shorthand~$\overline{S} := V \setminus S$.
Given two disjoint sets $A,B$ of vertices, $E(A,B)$ is the number of edges with one endpoint in~$A$ and one endpoint in~$B$, and~$E(A)$ is the number of edges with both endpoints in~$A$.

\subsubsection*{The Structure of the Article}
In Sect.~\ref{section_mindeg_arbitrary_thresholds} we show that in order to guarantee the existence of a contagious set of size~$k$ for thresholds~$k$ a minimum degree of at least~$\frac{k-1}{k} \cdot n$ is required.
Then we examine in Sect.~\ref{sec:speed} how fast the activation spreads in such graphs.
Our extremal results are given in Sect.~\ref{sec:extremal}: we provide tight bounds on the number of edges required to have a contagious set of size~$k$ if all thresholds are~$k$ or~$2$.
In Sect.~\ref{sec:ore} we examine Ore graphs as a natural generalization of Dirac graphs and prove the existence of a contagious set of size two in Sect.~\ref{section_proof_ore_graphs}.

\section{The Existence of Contagious Sets in Graphs of Bounded Minimum Degree}
\label{section_mindeg_arbitrary_thresholds}
We study what minimum degree guarantees the existence of a contagious set of size~$k$ if all thresholds are equal to~$k$.
Note that for~$k=1$ connectivity is required to ensure the existence of a contagious set of size one; connectivity is only guaranteed for a minimum degree of (at least)~$\frac{n}{2}$.

In the following we show that for an arbitrary~$k \geq 2$ a minimum degree of at least~$\frac{k-1}{k} \cdot n$ suffices. 
%
%Recall that is dependency on the minimum degree is best possible in the following sense: let~$k=2$, then a minimum degree of~$\frac{k-1}{k} \cdot n -1 = \frac{n}{2} -1$ is not sufficient. To see this, we consider the graph that consists of two disjoint cliques on exactly~$\frac{n}{2}$ nodes each. This graph does not have a contagious set of size two if all thresholds are two.
%
Example~\ref{example_generalized_mindegree_tight} at the end of this section demonstrates that a minimum degree of~$\left\lfloor \frac{k-1}{k} \cdot n \right\rfloor$ does not suffice for any even~$n$. 
Thus, our condition on the minimum degree is necessary and sufficient.
We show the following.
\begin{theorem}
	\label{theorem_generalization_mindegree_arbitrary_thresholds}
	For~$k \geq 2$ let~$G = (V,E)$ be a undirected graph on~$n$ nodes with minimum degree~$\left\lceil \frac{k-1}{k} \cdot n \right\rceil$. Then~$G$ has a contagious set of size~$k$.
\end{theorem}
\begin{proof}
	Assume for the sake of a contradiction that the graph has no contagious set of size~$k$. Then for every set~$S \subset V$ with~$|S| = k$ there is a set~$T$ with~$S \subseteq T \subset V$ such that the infection does not spread from the nodes in~$T$ to any node in~$\overline{T} = V \setminus T$. 
	Thus, it must be the case that every node in~$\overline{T}$ has at most~$k-1$ neighbors in~$T$.
	
	On the other hand, each node in~$T$ has at least~$\frac{k-1}{k} \cdot n - (|T| - 1)$ neighbors in~$\overline{T}$ due to the degree constraints.
	Thus, we will derive the desired contradiction (and may conclude that there must be a node in~$\overline{T}$ that becomes infected) if
	\begin{equation*}
		\left(\frac{k-1}{k} \cdot n - (|T| - 1)\right) \cdot |T| > (n - |T|) \cdot (k-1).
	\end{equation*}
	The next lemma narrows down the number of cases when the infection spreading may stop.
	\begin{lemma}
		\label{lemma_arbitrary_thresholds_stop}
		\label{lem:three} % this label is used later on
		For~$n \geq |T| \geq k \geq 2$ the infection does not stop unless~$|T| = k$ or~$|T| \geq \frac{k-1}{k} \cdot n$.
	\end{lemma}
	\begin{proof}
		%First observe that we may assume that~$n > \frac{k^2}{k-1}$ holds. 
		%
		First observe that
		\begin{align*}
			\left(\frac{k-1}{k} \cdot n - (|T| - 1)\right) \cdot |T| & > (n - |T|) \cdot (k-1)\\
			\intertext{holds if and only if}
			(k - |T|) \cdot \left(k\cdot |T| - (k-1) \cdot n\right) & > 0.
		\end{align*}
		Observe that the LHS is a quadratic function in~$|T|$,
		%$$f(|T|) = (k - |T|) \cdot \left(k\cdot |T| - (k-1) \cdot n\right),$$
		%
		with roots~$|T| = k$ and~$|T| = \frac{k-1}{k} \cdot n$.
		In particular, this function is non-positive in the case~$|T| \geq k$ only if~$|T| \geq \frac{k-1}{k} \cdot n$. 
	\end{proof}
	%
	%We point out that in particular every selection of three nodes is contagious for the case that all thresholds are two.
	%
	Therefore, it suffices to show that the infection process neither stops in the very first round when~$|T| = k$, nor late when the number of active nodes is already~$|T| \geq \frac{k-1}{k} \cdot n$.
	The following lemma will essentially settle the case~$|T| = k$.
	\begin{lemma}
		\label{lemma_arbitrary_thresholds_neighbors}
		Let~$n > k \geq 2$.
		For minimum degree~$\left\lceil \frac{k-1}{k} \cdot n \right\rceil$ every node has at least~$k$ neighbors.
	\end{lemma}
	\begin{proof}
		We begin by observing that
		\begin{align*}
			\frac{k-1}{k}\cdot n & \geq k\\
			\intertext{holds if and only if}
			%n & \geq \frac{(k+1)\cdot(k-1)+1}{k-1},\\
			%\intertext{or equivalently}
			n & \geq k+1 + 1/(k-1),
		\end{align*}
		hence for~$k \leq n-2$ every node has at least~$k$ neighbors as claimed.
		Therefore, it remains to examine the case~$k = n-1$:
		then by the integrality of degrees we have that
		\begin{align*}
			%\left\lceil\frac{k-1}{k}\cdot n\right\rceil & \geq k \nonumber\\
			\left\lceil\frac{n-2}{n-1}\cdot n\right\rceil & = \left\lceil\frac{n-2}{n-1}\cdot (n-1+1)\right\rceil\\
			& = n-2 + \left\lceil\frac{n-2}{n-1}\right\rceil,
		\end{align*}
		which equals~$n-1$ for~$n \geq 3$, as required; recall that~$n > k \geq 2$ holds.
		%
		%\qed
	\end{proof}
	Thus, we can always pick a node and activate $k$ of its neighbors to ensure that the process does not stop in the first step.
	We assume~$|T| \geq k+1$ from now on.
	
	Next we show that the infection process does not come to a standstill before all nodes are infected. 
	We distinguish the cases~$k=2$ and~$k \geq 3$.

	\paragraph*{The case~$k=2$.}
	Recall that in this case the minimum degree is at least~$\frac{n}{2}$, i.e.\ we consider Dirac graphs.
	We show the following simple, yet very helpful statement.
	\begin{lemma}
		\label{lemma_Dirac_almost_there}
		Let~$G = (V,E)$ be a Dirac graph and assume that every node has threshold two. If more than~$\frac{n}{2}$ nodes are active, then in the next round all remaining nodes become infected.
	\end{lemma}
	\begin{proof}
		Let~$A \subset V$ with~$|A| > \frac{n}{2}$ denote a set of active, resp.\ infected nodes. Then every node in~$\overline{A}$ can have at most~$|\overline{A}| - 1 \leq \left\lceil\frac{n}{2}\right\rceil - 2$ neighbors outside~$A$. Thus, it must have two neighbors in~$A$, since its degree is at least~$\left\lceil\frac{n}{2}\right\rceil$.
		%
		%\qed
	\end{proof}
	Lemma~\ref{lemma_arbitrary_thresholds_stop} gives that the number of active nodes is at least~$|T| \geq \frac{k-1}{k} \cdot n = \frac{n}{2}$ in this case, before the process may stop.
	Note that Lemma~\ref{lemma_Dirac_almost_there} guarantees that all nodes will become active once~$|T| > \frac{n}{2}$, therefore we wonder if the process can stop at~$|T| = \frac{n}{2}$?
	
	Interestingly, this is indeed the case. In order to make the process stop for threshold two, each of the~$\frac{n}{2}$ inactive nodes must have at most one active neighbor. Taking into account also the degree constraints, we observe that each inactive node has \emph{exactly} one active and~$\frac{n}{2} - 1$ inactive neighbors. For the active nodes it is the other way around.
	Therefore, the active and the inactive nodes form a clique of size~$\frac{n}{2}$ each, and both cliques are connected via a perfect matching.
	Note that the described scenario requires~$n$ to be even.
	
	We call this graph~$DC_n$, and observe that~$DC_n$ has a contagious set of size two (for thresholds two): simply activate one node in each clique such that the two activated nodes are not adjacent.
	Then the whole graph will become active in two rounds.
	
	Moreover, we have argued that if the graph is not~$DC_n$, then the process does not stop at~$|T| = \frac{n}{2}$ and will eventually activate all nodes.
	This proves the theorem for~$k=2$.

	\paragraph*{The case $k \geq 3$.}
	Recall that we may assume due to Lemma~\ref{lemma_arbitrary_thresholds_stop} that~$|T| \geq \frac{k-1}{k} \cdot n$ holds when the process stops. 
	%
	%Observe that this implies~$|T| > \frac{n}{2}$ active nodes at this point. 
	%
	We show a statement similar to Lemma~\ref{lemma_Dirac_almost_there}: the difference is that here we require \emph{at least}~$\frac{k-1}{k} \cdot n$ active nodes, whereas Lemma~\ref{lemma_Dirac_almost_there} required \emph{more than}~$\frac{n}{2}$ ones. 
	This difference will become important in Sect.~\ref{sec:speed}.
	\begin{lemma}
		\label{lemma_arbitrary_thresholds_almost_there}
		Assume that all thresholds are~$k$ and each node has minimum degree at least~$\frac{k-1}{k} \cdot n$.
		If at least~$\frac{k-1}{k} \cdot n$ nodes are active, then in the next round all nodes will become active.
	\end{lemma}
	\begin{proof}
		We have at most~$\left\lceil \frac{k-1}{k} \cdot n \right\rceil = \left\lfloor n - \frac{k-1}{k} \cdot n\right\rfloor = \left\lfloor \frac{n}{k}\right\rfloor$ inactive nodes. 
		Due to the degree constraints, the number of active neighbors of each inactive node is at least
		\begin{align*}
			\left\lceil\frac{k-1}{k}\cdot n\right\rceil - \left(\left\lfloor \frac{n}{k}\right\rfloor - 1\right) & = \left\lceil\frac{k-1}{k}\cdot n\right\rceil - \left(n - \left\lceil\frac{k-1}{k}\cdot n\right\rceil - 1\right)\\
			& = 2 \cdot \left\lceil\frac{k-1}{k}\cdot n\right\rceil -n +1.
		\end{align*}
		It suffices to show that this number is at least~$k$, since then every inactive node will be activated in the subsequent round.
		Therefore, we will show
		\begin{equation}
			\label{Iq_arb_thresholds}
			2 \cdot \left\lceil\frac{k-1}{k}\cdot n\right\rceil -n +1 \geq k. %\\
			%2 \cdot \left\lceil n - \frac{1}{k}\cdot n\right\rceil -n +1 &\geq k\\
			%n - 2\cdot \left\lfloor\frac{n}{k}\right\rfloor +1 \geq k
		\end{equation}
		Note that Eq.~(\ref{Iq_arb_thresholds}) is implied by
		\begin{align}
			2 \cdot \frac{k-1}{k}\cdot n -n +1 & \geq k \nonumber\\
			%\frac{k-2}{k}\cdot n +1 &\geq k \nonumber\\
			\intertext{or equivalenty,}
			(k-2) \cdot n & \geq k \cdot (k-1). \label{Eq_generalization_sufficient_neighbors}
		\end{align}
		Observe that Eq.~(\ref{Eq_generalization_sufficient_neighbors}) holds if~$n \geq k+2$ and~$k\geq 4$.
		Moreover, if~$k=3$ then Eq.~(\ref{Eq_generalization_sufficient_neighbors}) holds for~$n \geq k+3$. However, if~$k=3$ and~$n = k+2 = 5$, then~$\left\lceil\frac{k-1}{k}\cdot n\right\rceil = 4$ and the statement of the lemma holds trivially.
		
		Therefore, we focus on the case~$n = k+1$ and~$k \geq 3$.
		Then
		\begin{align*}
			2 \cdot \left\lceil\frac{k-1}{k}\cdot n\right\rceil -n +1 & = 2 \cdot \left\lceil\frac{(k-1)\cdot (k+1)}{k}\right\rceil -(k+1) +1\\
			%& = 2 \cdot \left\lceil\frac{k^2 -1}{k}\right\rceil -k \\
			& = 2 \cdot \left\lceil k - \frac{1}{k}\right\rceil  - k
		\end{align*}
		is at least~$k$ as desired, and inequality~(\ref{Iq_arb_thresholds}) holds for all~$n$ and~$k$ with~$n > k \geq 3$.
	\end{proof}
	Thus, the process does not stop if~$|T| \geq \frac{k-1}{k} \cdot n$ for~$k\geq 3$, and in particular all nodes will be active in the next round.
	The statement of the theorem follows since we already showed that it holds for~$k=2$.
	%
	%\qed
\end{proof}
We point out that in order to achieve an activation of the whole graph, it is sufficient to pick an arbitrary node and activate~$k$ of its neighbors (that must exist due to Lemma~\ref{lemma_arbitrary_thresholds_neighbors}).
Note that in particular this gives a contagious set of size two for Dirac graphs.

Next we show that our requirement on the minimum degree is optimal. Indeed, if lowering it to~$\left\lfloor \frac{k-1}{k} \cdot n \right\rfloor$, then the existence of a contagious set of size~$k$ is not guaranteed.
\begin{example}
	\label{example_generalized_mindegree_tight}
	For a clique on~$n$ nodes with~$n$ even, we pick an arbitrary perfect matching~$M$ and remove the edges of~$M$. Call the resulting graph~$G$ and let~$k = n-1$.
	Now we recall from the proof of Lemma~\ref{lemma_arbitrary_thresholds_neighbors} that
	\begin{align*}
		\left\lfloor \frac{k-1}{k} \cdot n \right\rfloor & = \left\lfloor \frac{n-2}{n-1} \cdot n \right\rfloor\\
		& = n-2,
	\end{align*}
	which equals the degree of the nodes in~$G$. Thus, the constraint on the minimum degree is met.
	
	Now observe that for thresholds~$k = n-1$ no inactive node can be activated by its neighbors, thus there is no contagious set of size~$k$ in~$G$.
	%
	% % % MP: For n odd, this statement is not necessarily true: let n=3 and k=2, then the min degree is 1 and we have a contagious set of size 2.
\end{example}

\section{The Speed of Spreading}
\label{sec:speed}
%\subsection{The Speed of Spreading for Thresholds~$k \geq 3$}
%
In Theorem~\ref{theorem_generalization_mindegree_arbitrary_thresholds} we proved for every~$k\geq 2$ that all graphs with minimum degree at least~$\frac{k-1}{k} \cdot n$ and thresholds~$k$ have a contagious set of size~$k$.

In this section we examine the question how many rounds the process requires to infect all nodes.
In the case of $DC_n$ and thresholds two, it is easy to see that any contagious set actually infects the entire graph in just two rounds. 
For graphs with minimum degree at least~$\frac{k-1}{k} \cdot n$ we will see that any contagious set of such a graph will infect all nodes within \emph{three} rounds if~$k \geq 3$.
In the case that~$k=2$, the process may take one more round, i.e.\ we show that any contagious set will activate the whole graph after at most \emph{four} rounds.
At the end of the section we give examples that prove these bounds tight.

In quick passing we point out that for~$k=1$ we required a minimum degree of at least~$\frac{n}{2}$ to ensure the existence of a contagious set of size one. Then we have at least~$\frac{n}{2} +1$ active nodes after the first round; hence every inactive node must have an active neighbor and thus will be infected after the second round.
Observe that this is tight for~$DC_n$ with~$n \geq 4$.
Now we show:
%
%%TODO Please read the new proof carefully
%
\begin{theorem}
	\label{theorem_speed_generalized}
	Let~$k \geq 2$ and denote by~$G = (V,E)$ a graph with minimum degree at least~$\frac{k-1}{k} \cdot n$ and all thresholds equal to~$k$. Then every contagious set activates the whole graph in at most three rounds if~$k \geq 3$, and in at most four rounds if~$k = 2$.
\end{theorem}
\begin{proof}
	When studying the existence of contagious sets for~$k=2$, we saw that the graphs called~$DC_n$ require a special treatment; for the same reason we exclude the case that~$k=2$ and the input is~$DC_n$ in the sequel. Recall that every contagious set in these graphs activates all nodes within two rounds, therefore this does not impair the generality of the statement of the theorem.
	
	Let~$A_0 \subset V$ be an arbitrary contagious set for~$G$, and recall that~$A_i$ is the set of all nodes active after round~$i \geq 1$ and moreover~$\overline{A_i} = V \setminus A_i$.
	We prove that the initial activation of \emph{any choice} of~$k+1$ nodes will infect the whole graph within the two rounds (for~$k \geq 3$) or three rounds respectively ($k = 2$); note that this is sufficient to prove the theorem since any contagious set (smaller than~$n$) must activate at least one node in the first round. 
	Therefore, we set~$A_1 = k+1$ in what follows.
	
	We showed in the proof of Theorem~\ref{theorem_generalization_mindegree_arbitrary_thresholds} that the process will infect all remaining nodes, once a critical mass of nodes is active:
	for~$k \geq 3$ Lemma~\ref{lemma_arbitrary_thresholds_almost_there} states that \emph{all} remaining nodes will be activated in the subsequent round, when~$\frac{k-1}{k} \cdot n$ nodes are active.
	
	For the case~$k=2$ we have that all nodes will be active after two more rounds if~$\frac{n}{2}$ nodes are active: since the graph is assumed not to be~$DC_n$, the process does not stop at~$\frac{n}{2}$ nodes (cp.\ the proof of Theorem~\ref{theorem_generalization_mindegree_arbitrary_thresholds}), and hence the number of active nodes will increase to at least~$\frac{n}{2} + 1$ in the next round. Therefore, all nodes will be active after another round due to Lemma~\ref{lemma_Dirac_almost_there}.
	Thus, we assume
	\begin{equation}
		\label{Eq_speed_general_1}
		|A_2| \leq \left\lceil\frac{k-1}{k} \cdot n\right\rceil -1
	\end{equation}
	from now on.
	
	Due to the degree constraints, every node in~$A_1$ has at least~$\left\lceil\frac{k-1}{k} \cdot n\right\rceil - \left(|A_1| - 1\right)$ neighbors in~$\overline{A}_1$ and we have
	\begin{equation}
		\label{Eq_speed_general_lb}
		|E(A_1, \overline{A}_1)| \geq |A_1| \cdot \left(\left\lceil\frac{k-1}{k} \cdot n\right\rceil - |A_1| + 1\right).
	\end{equation}
	
	On the other hand, each node in~$\overline{A}_1$ has at most~$k-1$ neighbors in~$A_0$, otherwise the respective node would have been infected after the first round.
	Moreover, due to Eq.~(\ref{Eq_speed_general_1}) we have that there are at most~$\left\lceil\frac{k-1}{k} \cdot n\right\rceil -1 - |A_1|$ nodes in~$\overline{A_1}$ that become active in the second round, i.e.\ they belong to~$A_2 \setminus A_1$: observe that these nodes can be adjacent to all the nodes in~$A_1 \setminus A_0$ but only to at most~$k-1$ nodes in~$A_0$, since otherwise they would belong to~$A_1$.
	Therefore, we obtain the following upper bound:
	\begin{equation}
		\label{Eq_speed_general_ub}
		|E(A_1, \overline{A}_1)| \leq \left(n - |A_1|\right) \cdot (k-1) + \left(\left\lceil\frac{k-1}{k} \cdot n\right\rceil -1 - |A_1|\right) \cdot \left(|A_1| - |A_0|\right).
	\end{equation}
	
	We will show momentarily that this upper bound is smaller than the lower bound given in Eq.~(\ref{Eq_speed_general_lb}). Therefore, we obtain a contradiction to the assumption in Eq.~(\ref{Eq_speed_general_1}): there must be at least~$\left\lceil\frac{k-1}{k} \cdot n\right\rceil$ active nodes after the second round, and hence all nodes will be active after the third round if~$k\geq 3$ and after the fourth round if~$k=2$ respectively. 
	
	Subtracting the RHS of Eq.~(\ref{Eq_speed_general_ub}) from the RHS of Eq.~(\ref{Eq_speed_general_lb}) gives
	\begin{align*}
		& |A_1| \cdot \left[\left\lceil\frac{k-1}{k} \cdot n\right\rceil - |A_1| + 1\right] \\
		& - \left(n - |A_1|\right) \cdot (k-1) - \left(\left\lceil\frac{k-1}{k} \cdot n\right\rceil -1 - |A_1|\right) \cdot \left(|A_1| - |A_0|\right)\\
		= & \left(k+1 - |A_0|\right) \cdot |A_1| + |A_0| \cdot \left(\left\lceil\frac{k-1}{k} \cdot n\right\rceil -1\right) - n \cdot (k-1).
		\intertext{Now we use~$|A_0| = k$ and $\left\lceil\frac{k-1}{k} \cdot n\right\rceil = n - \left\lfloor \frac{n}{k} \right\rfloor$ to obtain}
		= & |A_1| + k \cdot n - k \cdot \left\lfloor \frac{n}{k} \right\rfloor - k - k \cdot n + n\\
		\geq & |A_1| - k,
	\end{align*}
	which is at least one, since~$|A_1| = k+1$ holds by assumption.
\end{proof}
%
%%TODO either Dani's example goes here or replaces the above one
%
Next we show that our bound is tight for~$k=2$. Subsequently, we will give tight examples for larger values of~$k$.
% % MP: This example is redundant, since the example for k >= 4 also works for k=2.
%
\begin{example}
	\label{example_speed_tight_Dirac}
	Consider the following graph on the vertex set $V=\{v_1,\ldots,v_8\}$:
	let $v_1,v_2,v_3$ be a clique, $v_4$ and $v_5$ be adjacent to $v_1$, while $v_7$ and $v_8$ are adjacent to $v_2$. Moreover, let $v_3$ be adjacent to $v_4$ and $v_6$. $v_4$ and $v_5$ are adjacent to each other as well as to $v_7$ and $v_8$. $v_7$ and $v_8$ are adjacent to each other and to $v_6$. $v_6$ is also adjacent to $v_5$. Every vertex has degree at least four. Thus, if~$v_1$ and~$v_2$ are activated, it takes four rounds for the entire graph to be infected. 
	Moreover, if we activate any three nodes, then it takes three rounds for the entire graph to be infected.
\end{example}
%
% % % MP: Proof seems incorrect.
%\begin{example}
%\label{example_speed_tight_Dirac_large}
%Consider for odd $k$ a graph on the vertex set $V=\{v_0,\ldots,v_{k^3-1}\}$. The edge-set is defined by $E=\{(v_i,v_j) : |i-j|>\frac{k-1}{2} \mod k^3\}$. Notice that each vertex is non-adjacent only to itself and $2\frac{k-1}{2}$ additional vertices, implying that each vertex has degree $k^3-(k-1)-1=\frac{k-1}{k}k^3=\frac{k-1}{k}n$.
%
%We claim that with $A_0=\{v_{j k^2}|j=0,\ldots,(k-2)\}\cup \{v_{(k-1)k^2-1}\}$ it takes three rounds for all vertices to become infected.
%
%Notice that in the first round, all vertices other than $v_{k^3-\frac{k-1}{2}-1}$ are non-adjacent to one of the $k$ activated vertices. Thus, $A_1=A_0 \cup \{v_{k^3-\frac{k-1}{2}-1}\}$. What is now left to show is the existence of at least one vertex $v_j$ such that $v_j\not\in A_2$. To do so, we only need to show that there exists a vertex $j$ that is not in $A_0$ and non-adjacent to $v_{k^3-\frac{k-1}{2}-1}$. Vertices $v_j$ with \[j\in\{k^3-k,\ldots,k^3-\frac{k-1}{2}-2\}\cup\{k^3-\frac{k-1}{2}\ldots,v_{k^3-1}\}\] fulfill those properties.
%\end{example}
%%TODO add corrected example for k > 2
%
The next example shows the upper bound tight for odd $k \geq 3$.
\begin{example}
	\label{example_speed_tight_Dirac_large}
	Consider for odd $k \geq 3$ a graph on the vertex set $V=\{v_0,\ldots,v_{k^3-1}\}$. The edge set is defined by $E=\{(v_i,v_j) \; \mid \; |i-j|>\frac{k^2-1}{2} \mod k^3\}$. Notice that each vertex is non-adjacent only to itself and $2\cdot \frac{k^2-1}{2}$ additional vertices, implying that each vertex has degree $k^3-(k^2-1)-1=\frac{k-1}{k} \cdot k^3=\frac{k-1}{k} \cdot n$.
	We show that the activation process takes three rounds to infect all nodes when starting with the following seed: $A_0=\{v_{j \cdot k^2} \;|\; j=0,\ldots,(k-2)\}\cup \{v_{(k-1)\cdot k^2-1}\}$.
	
	Notice that at the beginning of the first round all inactive vertices are non-adjacent to one of the vertices in~$A_0$ except for the node~$v_{k^3-\frac{k^2-1}{2}-1}$: to see this, observe that in particular~$\left(k^3-\frac{k^2-1}{2}-1\right) - \left((k-1)\cdot k^2-1\right) = \frac{k^2 + 1}{2}$ and $k^3 - \left(k^3-\frac{k^2-1}{2}-1\right)$ are both larger than~$\frac{k^2-1}{2}$.
	Thus, $A_1=A_0 \cup \{v_{k^3-\frac{k^2-1}{2}-1}\}$. 
	
	Now we show that there is a vertex that is still inactive at the end of the second round, and hence by Theorem~\ref{theorem_speed_generalized} will be infected in the third.
	Observe that the nodes~$v_j$ with 
	$j\in\{k^3-k^2,\ldots,k^3-\frac{k^2-1}{2}-2\}\cup\{k^3-\frac{k^2-1}{2},\ldots,k^3-1\}$
	do not belong to~$A_0$ and are not adjacent to~$v_{k^3-\frac{k^2-1}{2}-1}$, the node infected in the first round. Thus, by our observation above, they do not have~$k$ neighbors in~$A_1$.
\end{example}
%
% % % MP: Dani's revised version below. It is correct to my best.
%\begin{example}
%	\label{example_speed_tight_Dirac_large}
%	Consider for odd $k$ a graph on the vertex set $V=\{v_0,\cdots,v_{k^3-1}\}$. The edge-set is defined by $E=\{(v_i,v_j) : |i-j|>\frac{k^2-1}{2} \mod k^3\}$. Notice that each vertex is non-adjacent only to itself and $2\frac{k^2-1}{2}$ additional vertices, implying that each vertex has degree $k^3-(k^2-1)-1=\frac{k-1}{k}k^3=\frac{k-1}{k}n$.
%	
%	We claim that with $A_0=\{v_{j k^2}|j=0,\cdots,(k-2)\}\cup \{v_{(k-1)k^2-1}\}$ it takes three rounds for all vertices to become infected.
%	
%	Notice that in the first round, all vertices other than $v_{k^3-\frac{k^2-1}{2}-1}$ are non-adjacent to one of the $k$ activated vertices. Thus, $A_1=A_0 \cup \{v_{k^3-\frac{k^2-1}{2}-1}\}$. What is now left to show is the existence of at least one vertex $v_j$ such that $v_j\not\in A_2$. To do so, we only need to show that there exists a vertex $j$ that is not in $A_0$ and non-adjacent to $v_{k^3-\frac{k^2-1}{2}-1}$. Vertices $v_j$ with \[j\in\{k^3-k^2,\cdots,k^3-\frac{k^2-1}{2}-2\}\cup\{k^3-\frac{k^2-1}{2}\cdots,v_{k^3-1}\}\] fulfill those properties.
%	
%	%But then, vertex $v_{k^3-\frac{k-1}{2}-1}$
%	%
%	%
%	%\qed
%\end{example}
%
The last example addresses even~$k \geq 4$.
\begin{example}
	\label{example_speed_tight_Dirac_even}
	Consider a graph with vertex set $V=\{v_0,\ldots,v_{k\cdot (k+1)-1}\}$ for even $k$. 
	Then the condition on the minimum degree is $\frac{k-1}{k}\cdot n=(k-1)\cdot (k+1)=k^2-1$, so each vertex must be adjacent to at least $k^2-1$ vertices, i.e.\ it can be non-adjacent to at most~$k$ vertices (other than itself). 
	Therefore, we let $E=\{(v_i,v_j) \; \mid \; |i-j|>\frac{k}{2}\mod k^2+k\}$. For the seed that contains the vertices $v_{j \cdot (k+1)}$ with $j\in\{0,\ldots,k-2\}$ as well as vertex $v_{(k-1)\cdot (k+1)-1}$, we observe analogously to Example~\ref{example_speed_tight_Dirac_large} that only vertex $v_{k^2+\frac{k}{2}-1}$ becomes infected in the first round. 
	Moreover, in the second round the vertices whose indices differ by at most~$\frac{k}{2}$ from~$k^2+\frac{k}{2}-1$ do not become infected, since they are not adjacent to~$v_{k^2+\frac{k}{2}-1}$.
	Thus, Theorem~\ref{theorem_speed_generalized} states that these nodes will be infected in the third round since~$k > 2$.
\end{example}
%
%\begin{example}
%	\label{example_speed_tight_Dirac_even}
%	Consider a graph with vertex set $V=\{v_0,\cdots,v_{k(k+1)-1}\}$ for even $k$. Then $\frac{k-1}{k}n=(k-1)(k+1)=k^2-1$, so each vertex must be adjacent to at least $k^2-1$ vertices, i.e. it can be non-adjacent to $k$ vertices other than itself. Similarly to example \ref{example_speed_tight_Dirac_large}, set $E=\{(v_i,v_j) : |i-j|>\frac{k}{2}\mod k^2+k\}$. Activating vertices $v_{j(k+1)}$ for $j\in\{0,\cdots,k-2\}$ as well as vertex $v_{(k-1)(k+1)-1}$, we find that only vertex $v_{k^2+\frac{k}{2}-1}$ becomes infected in the first round. But then, in the second round vertices with index within $\frac{k}{2}$ of $k^2+\frac{k}{2}-1$ do not become infected, so in this case as well, it takes three rounds for the whole graph to become infected.
%\end{example}

\section{Extremal Number of Edges}
\label{sec:extremal}
%We will use the following lemma to prove a tight bound on the number of edges that guarantees the existence of a contagious set of size $k\ll n$. We did not attempt to find the exact $k(n)$ for which Theorem~\ref{thm:R} holds. It should be noted that certain restrictions on $k$ have to be imposed in order for Theorem~\ref{thm:R} to be valid, as is shown in Example~\ref{example_Mnk2}.
%%
%\input{extremal1}
%
%\noindent{}We now turn our attention to bounds on $M(n,k,k)$.
%
What is the maximum number of edges in a graph with~$n$ nodes such that there is no contagious set of size~$k$ assuming that all nodes have threshold~$k$?
We provide the following tight bound for the case~$n\geq 2k+2$.
\begin{theorem}
	\label{theorem_extremal_mkk}
	Let~$k \geq 1$.
	For $n\geq 2k+2$ we have $M(n,k,k)=\binom{n-1}{2}$.
\end{theorem}
\begin{proof}
	To see~$M(n,k,k)\geq\binom{n-1}{2}$, note that a clique of~$n-1$ nodes plus an isolated node is a disconnected graph with~$n$ nodes and~$\binom{n-1}{2}$ edges. However, no disconnected graph can have a contagious set of size $k<n$ when the thresholds are $k$.
	In the sequel we show~$M(n,k,k)\leq\binom{n-1}{2}$, i.e.\ every graph on~$n$ nodes with at least~$\binom{n-1}{2} + 1$ edges has a contagious set of size~$k$ if all thresholds are~$k$.
	
	If a set~$S \subseteq V$ with~$|S| = k$ is not contagious for all thresholds equal to~$k$, then there is a set~$T$ with~$S \subseteq T$ such that each node in~$\overline{T}$ has at most~$k-1$ neighbors in~$T$; only then the infection does not spread outside of~$T$.
	
	The gist is that there are at least $$|\overline{T}| \cdot \left(|T| - (k-1)\right) = \left(n - |T|\right) \cdot \left(|T| - (k-1)\right)$$ pairs of nodes in~$T \times \overline{T}$ that are not adjacent.
	In particular, we claim that if~$|T| \in \{k+1, k+2,..., n-3, n-2\}$ then the number of non-adjacent node pairs is larger than~$n-2$.
	However, at most~$n-2$ pairs of nodes are not adjacent in a graph with~$n$ nodes and at least~$\binom{n-1}{2} + 1$ edges, since~$\binom{n-1}{2} + 1 = \binom{n}{2} - (n-2)$
	holds.
	Thus, if the number of active vertices is at least~$k+1$ and at most~$n-2$, then in the subsequent round at least one more node is infected newly, and therefore the infection does not stop until at least~$n-1$ nodes are active.
	
	Now we prove the claim. First observe that
	$$f(|T|) = \left(n - |T|\right) \cdot \left(|T| - (k-1)\right)$$
	is a quadratic function in~$|T|$.
	%and has roots~$|T|= n$ and~$|T| = k-1$. In the former case, the process has already infected all nodes, and the latter case cannot occur, since~$S \subseteq T$ and~$|S| = k$.
	%
	Next we show that~$f(|T|)$ is larger than~$n-2$ for values of $|T| \in \{k+1,\ldots,n-2\}$.
	Recall that we assumed~$n \geq 2k + 2$, and observe that the number of non-adjacent pairs in~$T \times \overline{T}$ is minimized for any fixed~$k$ by setting~$n = 2k + 2$. Therefore the number of such pairs is at least~$(2k+2 - |T|)\cdot (|T| - k + 1)$. On the one hand, if~$|T| = k+1$ holds, then their number is~$(k+1) \cdot 2 = 2k+2 = n$. On the other hand, for~$|T| = n-2 = 2k$ their number is~$2 \cdot (k+1) = n$ again. Thus, the claim holds for both values of~$|T|$, and furthermore for all choices of~$|T|$ in between, since~$f''(|T|) = -2$ and hence~$f$ is concave.
	
	Thus, we focus on~$|T| \in \{k,n-1\}$ in the sequel.
	First we show how to select~$A_0$ with~$|A_0| = k$ such that~$|A_1| \geq k+1$ holds.
	If the graph does not contain any node of degree less than~$k$, we pick any node~$v$ and choose~$A_0$ to contain~$k$ neighbors of~$v$. Then~$v \in A_1$ holds and hence~$|A_1| \geq k+1$.
	
	Now assume there is a node~$u$ with degree~$d < k$.
	Note that any node of degree smaller~$k$ is non-adjacent to at least~$n-1 - \frac{n-2}{2} = \frac{n}{2}$ nodes, where we use~$\frac{n-2}{2} \geq k$. Hence there can be at most one such node because there are at most~$n-2$ non-adjacent pairs of nodes in the graph.
	
	Let~$G'$ be the graph after removing~$u$ and its~$d$ incident edges. Note that the degree of each node in~$G'$ was reduced by at most one due to the removal of~$u$, therefore all degrees in~$G'$ are at least~$k-1$. Hence we pick any node~$w$ that was adjacent to~$u$ (recall that the graph is connected) and choose~$A_0$ to contain~$u$ and~$k-1$ neighbors of~$w$. Then in the first round~$w$ will be infected, thus we have~$|A_1| \geq k+1$.
	
	We have already shown that if there are at least~$k+1$ active nodes, then the process does not stop until there are~$n-1$ active nodes.
	The last node becomes infected unless its degree is less than~$k$, but in this case it would have been selected for~$A_0$.
	%Assume that the process does not infect the last node~$w$. Then~$w$ has degree less than~$k$; but in this case~$w$ was selected for~$A_0$.
	Thus, the process cannot stop at~$n-1$ nodes.
	%
	%\qed
\end{proof}
%
%We give two examples that illustrate the dependency of~$k$ on~$n$.
%
Note that the statement of Theorem~\ref{theorem_extremal_mkk} does not hold for arbitrary $k\in\{1,\ldots, n-1\}$:
indeed, we required~$n \geq 2k + 2$.
Next we consider the case that~$k = n-1$, i.e.~$k$ is very large compared to~$n$.
\begin{theorem}
	For all~$n \geq 2$ we have $M(n,n-1,n-1) = \binom{n}{2} - \left\lceil\frac{n}{2}\right\rceil$ holds.
\end{theorem}
\begin{proof}
	We reuse the graph of Example~\ref{example_generalized_mindegree_tight} in Sect.~\ref{section_mindeg_arbitrary_thresholds}.
	For a clique on $n$ vertices we pick a perfect matching~$M$ if~$n$ is even. For odd~$n$ we set~$M$ to be a near-perfect matching plus an edge incident to the unmatched node.
	Then we delete the edges of~$M$. Let~$k=n{-}1$. Each vertex has degree~$n{-}2$, hence there is no contagious set of size $n{-}1$.
	Thus, $M(n,n-1,n-1) \geq \binom{n}{2} - \left\lceil\frac{n}{2}\right\rceil$.
	
	On the other hand, if we add one more edge, i.e.\ if we require the graph to have at least~$\binom{n}{2} - \left\lceil\frac{n}{2}\right\rceil +1$ edges, then there must be a node~$u$ that is adjacent to all others: there are at most~$\left\lceil\frac{n}{2}\right\rceil -1$ edges not present in the graph (compared to~$K_n$). Hence activating~$V \setminus \{u\}$ is a contagious set and we have shown that $M(n,n-1,n-1) = \binom{n}{2} - \left\lceil\frac{n}{2}\right\rceil$ holds.
	%
	%\qed
\end{proof}
%
%\begin{example}
%Fix~$k=n-2$ and denote by~$A$ and~$B$ two vertex sets of size~$\frac{n}{2}$ each.
%Let~$A$ be a clique and~$B$ be a clique where the edges of a perfect matching were deleted. Between~$A$ and~$B$ all edges exist except for a perfect matching between~$A$ and~$B$ whose edges were deleted.
%%
%Then the graph has $\binom{n}{2}-\frac{3n}{4}= \binom{n-1}{2}+\frac{n}{4}-1$ edges. Yet, the vertices in $B$ have $\frac{n}{2}-2$ neighbors in $B$ and $\frac{n}{2}-1$ neighbors in $A$ giving degree $n-3$, so if there was a contagious set of size $n-2$, then it would have to contain~$B$. But then, only $\frac{n}{2}-2$ vertices from~$A$ could be in the contagious set and the remaining two nodes in~$A$ would each have only $n-3$ active neighbors, so they would not be infected.
%%
%\qed
%\end{example}

Now we focus on the case that all nodes have threshold two.
We will use the following lemma to prove a tight bound on the number of edges that guarantees the existence of a contagious set of size $k\ll n$. We did not attempt to find the exact $k(n)$ for which Theorem~\ref{thm:R} holds. It should be noted that certain restrictions on $k$ have to be imposed, as is shown in Example~\ref{example_Mnk2}.
\begin{lemma}\label{largeU}
	Consider an $n$-vertex graph $G$ with at least~$\binom{n-k+1}{2}+\lfloor \frac{k+1}{2}\rfloor$ edges and suppose that $n>32k+4$. Then there is an induced subgraph $G'$ of $G$, such that $|V(G')|\geq n-8k$ and each vertex in $G'$ has degree at least $\frac{|V(G')|}{2}+1$ within~$G'$.
\end{lemma}
%
%Due to space constraints the proof was moved to the appendix.
%
\begin{proof}
	%TODO this proof has to be revised. I suggest using c=1/4 from beginning on...
	By our assumptions, $G$ has at most $$\binom{n}{2}-\binom{n-k+1}{2}-\left\lfloor\frac{k+1}{2}\right\rfloor=(k-1)\cdot n-\frac{k^2-k}{2}-\left\lfloor\frac{k+1}{2}\right\rfloor<(n-1)\cdot k$$
	pairs of non-adjacent vertices. 
	Since in a complete graph each vertex has degree~$n-1$, at most $\frac{2k}{1/4}=8k$ vertices in~$G$ can have fewer than $(1 - \frac{1}{4})\cdot n-1$ neighbors, as the number of pairs of non-adjacent vertices is at most $nk$. We prove that the vertices with degree at least $(3/4)n$ in $G$ form a subgraph $G'$ with the desired property. Let $U$ consist of all vertices in $G$ of degree at least $(3/4)n$ and $G'$ be the graph induced by $U$. Each vertex in $U$ can be adjacent to at most $8k$ vertices in $\overline{U}$ since $|V\setminus U|\leq 8k$. It follows that every vertex in $U$ must be adjacent to at least $(3/4)n-8k$ vertices within $U$. Since $|U|\leq n$, it suffices for each vertex in $G'$ to have $\frac{n}{2}+1$ neighbors within $U$.
	%
	%\qed
\end{proof}

\begin{theorem}\label{thm:R}
	For all $k\geq 2$  there exists $n_k\in\mathbb{N}$, such that for all $n\geq n_k$,
	$$M(n,k,2)=\binom{n-k+1}{2}+\left\lfloor\frac{k-1}{2}\right\rfloor.$$
\end{theorem}
\begin{proof}
	We begin by lower bounding $M(n,k,2)$. Consider for odd $k$ a graph $G$ on $n$ vertices, where $k{-}1$ vertices form a perfect matching, i.e.\ there are $\frac{k-1}{2}$ isolated edges, and the remaining $n{-}k{+}1$ vertices form a clique. Observe~$|E|=\binom{n-k+1}{2} {+} \frac{k-1}{2}$. We claim $m(G,2) {>} k$: a contagious set would have to contain the $k{-}1$ former vertices plus two vertices in the clique (i.e., $k{+}1$ vertices). Notice that adding any edge would decrease the size of a minimum contagious set to $k$. For even $k$, a similar construction with an isolated vertex, $\frac{k-2}{2}$ isolated edges, and a clique on $n-k+1$ vertices shows the corresponding bound.
	
	We will prove the upper bound on~$M(n,k,2)$ for $n\geq n_k=9k{+}\binom{8k}{2}$. Consider a graph $G$ with at least $\binom{n-k+1}{2}{+}\lfloor\frac{k-1}{2}\rfloor {+}1$ edges. Define $U$ as a subset of vertices of maximum cardinality with the property that each vertex in $U$ has at least $\frac{|U|}{2}{+}1$ neighbors in $U$. By Lemma \ref{largeU}, $|U|\geq n{-}8k$. Let  $W{:=}\langle U \rangle {\setminus} U$ and $R=V{\setminus} (U{\cup} W)$. Notice then that $|R|\leq 8k$ and each vertex in $R$ can have at most one neighbor in $U{\cup} W$, i.e.\ $E(R,V{\setminus}R) {\leq} |R|$, as otherwise this vertex would have been infected by $U{\cup} W$ and thus would not belong to $R$.
	
	As any two vertices in $U$ share a common neighbor, by Lemma~\ref{lem:three}, all of $U$ is eventually infected when just two vertices in $U$ are active. In this case, $U$ and by definition also all of $W$ would become active. In the remainder of the proof, we will show by case distinction that with $n\geq n_k$ the assumption that there is no contagious set of size $k$ implies that the number of edges in the graph is bounded above by $\binom{n-k+1}{2}+\left\lfloor\frac{k-1}{2}\right\rfloor$. This contradicts our choice of $G$ as a graph with $\binom{n-k+1}{2}+\left\lfloor\frac{k-1}{2}\right\rfloor+1$ edges and thus proves that $G$ must have a contagious set of size $k$. The different cases we need to consider are as follows:
	\begin{itemize}
		\item The number of edges within $R$ is greater than $\left\lfloor\frac{|R|}{2}\right\rfloor$.
		\item The number of edges within $R$ is at most $\left\lfloor\frac{|R|}{2}\right\rfloor$ and there are no edges from $R$ to $U\cup W$.
		\item The number of edges within $R$ is at most $\left\lfloor\frac{|R|}{2}\right\rfloor$ and there is at least one edge from $R$ to $W$.
		\item There is at least one edge from $R$ to $U$.
		
	\end{itemize}
	
	We first consider the case that $E(R)>\left\lfloor\frac{|R|}{2}\right\rfloor$.  Then there must be a vertex $v \in R$ with at least two neighbors in $R$. Activating $R{\setminus}\{v\}$ will infect $v$. Thus,  there is a contagious set of size $|R|{-}1{+}2$ and if $m(G,2){>}k$ we must have $|R|{\geq} k$. 
	%
	%Suppose that $|R| \geq k$. 
	%
	Since the number of edges within $R$ is at most $\binom{|R|}{2}$ and we already saw $E(R,V {\setminus} R) {\leq} |R|$, we have that $|E| {\leq} \binom{n-|R|}{2}{+}|R|{+}\binom{|R|}{2}$.
	Assuming the bounds $k\leq |R|$, $n\geq n_k$ however, guarantees
	\begin{equation}
		\label{Eq_star}
		\binom{n-k+1}{2}-\binom{n-|R|}{2} \geq \binom{n-k+1}{2}-\binom{n-k}{2}\geq n-k
	\end{equation}
	%\begin{align}
	%	\binom{n-k+1}{2}-\binom{n-|R|}{2} &\geq \binom{n-k+1}{2}-\binom{n-k}{2}\geq n-k \label{Eq_star}\\
	%	\binom{n-k+1}{2}-n+k & \geq \binom{n-|R|}{2}. \nonumber
	%\end{align}
	%Therefore 
	%\begin{equation}
	%	\label{Eq_star}
	%	\binom{n-|R|}{2}\leq \binom{n-k+1}{2}-n+k.
	%\end{equation}
	Thus, we have~$\binom{n-k+1}{2}-n+k \geq \binom{n-|R|}{2}$ and
	\begin{align}\label{Eq_star2}
		|E| & \leq \binom{n-|R|}{2}+|R|+\binom{|R|}{2}\leq \binom{n-k+1}{2}-n+k+|R|+\binom{|R|}{2} \nonumber \\
		& \leq \binom{n-k+1}{2}.
	\end{align}
	where the second inequality follows from our observation that~$|R|\leq 8k$ and the initial assumption $n\geq n_k=9k+\binom{8k}{2}$.
	This settles the first case.

	Assume next that $E(R) {\leq} \left\lfloor\frac{|R|}{2}\right\rfloor$ and there are no edges from $R$ to $U{\cup} W$. In this case, $|E|{\leq}  \binom{n-|R|}{2} {+} \left\lfloor \frac{|R|}{2}\right\rfloor$. Suppose no contagious set of size $k$ exists, so $|R|{>}k{-}2$. Thus, $$|E| \leq \binom{n-|R|}{2}{+} \left\lfloor \frac{|R|}{2}\right\rfloor \leq \binom{n-k+1}{2} {+}\left\lfloor\frac{k-1}{2}\right\rfloor.$$
	%\begin{equation*}
	%	\binom{n-|R|}{2}+ \left\lfloor \frac{|R|}{2}\right\rfloor\leq \binom{n-k+1}{2} +\left\lfloor\frac{k-1}{2}\right\rfloor.
	%\end{equation*}
	This inequality is tight for $|R|=k{-}1$ and follows from Inequality~(\ref{Eq_star}) and from~$n{-}k > \left\lfloor \frac{|R|}{2}\right\rfloor$ if $|R|{>}k{-}1$. This concludes the second case.
	
	% % % MP: Here it said ``$E(R) \geq \left\lfloor \frac{|R|}{2}\right\rfloor$'' but I think it should be \leq
	%
	Suppose $E(R) \leq \left\lfloor \frac{|R|}{2}\right\rfloor$ and that there is an edge from a vertex in $R$ to a vertex $v\in W$. Recall that $U$ was chosen as the set of maximum cardinality with the property that every vertex $u\in U$ has at least $\frac{|U|}{2}{+}1$ neighbors. 
	Therefore, the graph $H$ induced on $U \cup \{v\}$ does not have the property that every vertex in $H$ has at least $\frac{|U|+1}{2}{+}1$ neighbors in $H$. Thus, there must be a vertex in $H$ that has less than $\frac{|U|+1}{2}{+}1$ neighbors within $H$. It follows that there are at least $\frac{|U|-1}{2}{-}1$ non-adjacent pairs of nodes in $H$. Recall that Lemma~\ref{largeU} implied that $|U| \geq n {-} 8k$. 
	Subtracting the non-existing edges inside~$U$ from the number of possible edges in~$V {\setminus} R$, we obtain
	\begin{align*}
		E(U \cup W) & \leq \binom{n-|R|}{2} - \left(\frac{|U| - 1}{2} - 1\right) = \binom{n-|R|}{2} - \frac{|U|}{2} + \frac{3}{2}\\
		& \leq \binom{n-|R|}{2} - \left(\frac{n}{2} - 4k - \frac{3}{2}\right).
	\end{align*}
	%
	% % % Explain where the term (n / 2 - 4k - 2) comes from (p. 10, line 4).
	%
	% % % In fact, the 2 was apparently wrong.
	%
	As there are at most~$|R|$ edges from $R$ to $W$ and at most $\left\lfloor\frac{|R|}{2}\right\rfloor$ within $R$, we find that
	\[
	|E|\leq \binom{n-|R|}{2}-\left(\frac{n}{2}-4k-\frac{3}{2}\right)+|R|+\left\lfloor \frac{|R|}{2}\right\rfloor\leq \binom{n-|R|}{2},\]
	where the inequality holds with $n\geq n_k$ since $|R|\leq 8k$. Activating all of $R$ and two vertices in $U$ gives a contagious set of size $|R|+2$. If no contagious set of size $k$ exists, it follows that $|R|> k-2$ and $|E|$ is upper bounded again by $\binom{n-k+1}{2}$.
	
	Finally, we deal with the case if there is an edge from $R$ to a vertex $v \in U$. By construction, activating all vertices in $R$ together with a neighbor of $v$ in $U$, must infect $v$, then a third vertex in $U$ and then  --- by Lemma~\ref{lem:three} --- all of~$U$. Thus, there is a contagious set of size $|R|+1$.
	If no contagious set of size $k$ exists, then $|R|\geq k$ and Inequality (\ref{Eq_star}) holds. Moreover, as $|E|$ in this case is bounded from above by $\binom{n-|R|}{2}+|R|+\binom{|R|}{2}$, Inequality (\ref{Eq_star2}) implies that $|E|\leq \binom{n-k+1}{2}$.
	
	We have thus exhausted all cases. In each one, the assumption that no contagious set of size $k$ exists contradicted $G$ having more than $\binom{n-k+1}{2}+\left\lfloor\frac{k-1}{2}\right\rfloor$ edges. Thus, we may conclude that if~$n\geq n_k$, then $M(n,k,2)=\binom{n-k+1}{2}+\left\lfloor\frac{k-1}{2}\right\rfloor$.
	%
	%\qed
\end{proof}
\begin{example}
	\label{example_Mnk2}
	We construct a family of graphs to demonstrate that $$M(n,k,2)=\binom{n-k+1}{2}+\left\lfloor\frac{k-1}{2}\right\rfloor$$ 
	does not hold for arbitrary $k$ and $n$. Consider for $k\geq\frac{2n+2}{3}$ a clique on $n-k$ vertices together with a star on $k$ vertices. All $k-1$ leaves of the star must be contained in a contagious set and so do two vertices from the clique, so there is no contagious set of size $k$. However, the number of edges is
	$\binom{n-k}{2}+(k-1)=\binom{n-k+1}{2}+\frac{k}{2}+\frac{3k}{2}-(n+1)
	\geq \binom{n-k+1}{2}+\frac{k}{2} > \binom{n-k+1}{2}+\left\lfloor\frac{k-1}{2}\right\rfloor.$
	%\begin{multline*}
	%\binom{n-k}{2}+(k-1)=\binom{n-k+1}{2}+\frac{k}{2}+\frac{3k}{2}-(n+1) \\
	%\geq \binom{n-k+1}{2}+\frac{k}{2} > \binom{n-k+1}{2}+\left\lfloor\frac{k-1}{2}\right\rfloor.
	%\end{multline*}
	%\qed
\end{example}

\section{Contagious Sets of Size Two in Ore Graphs}
%\label{sec:dirac}
\label{sec:ore}
In this section we study Ore graphs, that are a generalization of Dirac graphs.
Recall that an $n$-vertex graph is a Dirac graph if every vertex in the graph is of degree at least $n/2$. For an Ore graph~$G = (V,E)$ we have that~$u,v \in V$ with~$(u,v) \notin E$ implies~$\mathrm{deg}(u) + \mathrm{deg}(v) \geq n$.
We focus in this section exclusively on the case~$r=2$.

%\subsection{The Existence of Small Contagious Sets}\label{sec:dirac}
The upper bound in \cite{Ackerman,Reichman} shows that in a Dirac graph there \emph{exists} a contagious set of size three.
In Theorem~\ref{theorem_generalization_mindegree_arbitrary_thresholds} we showed that there even exists a contagious set of size two.

We wish to generalize this result to Ore graphs, however, some new ideas are required.
For example, Lemma~\ref{lemma_arbitrary_thresholds_stop} does not extend to Ore graphs: there exist $n$-vertex Ore graphs such that there is a selection of up to~$\left\lfloor\frac{n}{2}\right\rfloor$ nodes that do not form a contagious set.
\begin{example}
	\label{example_ore_graphs}
	We construct the graph as follows: the set $S = \{v_1,v_2,\ldots, v_c\}$ forms a clique. The remaining~$n-c$ nodes also form a clique, and are partitioned into~$c$ disjoint groups~$G_1, G_2,\ldots,  G_c$. 
	We let~$c \leq \lfloor\frac{n}{2}\rfloor$, thus every~$G_i$ is non-empty.
	Every node in~$G_i$ is adjacent to~$v_i$ but not to any other node in~$S$. Hence~$S$ is not a contagious set.
	Moreover, note that for any pair~$(u,v) \in S \times \overline{S}$ we have $$\mathrm{deg}(u) + \mathrm{deg}(v) = (c - 1 +1) + (n - c - 1 +1) = n,$$ hence we have constructed an Ore graph.
	Here it is crucial to note that pairs of nodes within~$S$ (and in~$\overline{S}$ resp.) are adjacent and hence their degrees are not required to sum up to~$n$ in a pairwise manner.
	Notice that for $c=\frac{n}{2}$, the constructed graph is $DC_{n}$.
	%
	%\qed
\end{example}
In Sect.~\ref{section_proof_ore_graphs} we show the following.
%\noindent{}Now we show the following.
%
\begin{theorem}
	\label{thm:Ore}
	Every Ore graph $G=(V,E)$ has a contagious set of size two.
\end{theorem}

\section{Conclusion}
We have examined conditions on the minimum degree and the average degree of undirected graphs ensuring the existence of contagious sets of size~$k \geq 2$.
Moreover, we have studied the speed with which the infection spreads through the graph.
Our bounds on the number of rounds that any contagious set of size~$k$ requires to activate the whole graphs are tight.

There are several questions that arise from this work. One is to determine the value of $M(n,k,r)$ for all $n \geq k \geq r$.
%
% % % MP: Removed this question, since I have solved the case $m(G,r)=r$. If we want to keep the question in, we should contrast it more.
%Another question is to find how large the minimum degree of~$G$ needs to be in order to ensure that~$m(G,r)=k$, where $k\geq r>2$.
%
Finally, it might be of interest to discover additional graph properties implying $m(G,k)=k$.

\paragraph*{Acknowledgments.}
The authors would like to thank the reviewers of the conference version of this paper for their valuable comments that helped improve the presentation significantly.

%\newpage

%\section*{References}

\newpage

\appendix

\section{The Proof of Theorem~\ref{thm:Ore}}
\label{section_proof_ore_graphs}

In this section we give the proof of Theorem~\ref{thm:Ore}. For the convenience of the reader, we restate the theorem.
%TODO adapt number if necessary

\medskip

\noindent\textbf{Theorem 15, restated.}
\textit{
	Every Ore graph $G=(V,E)$ has a contagious set of size two.
}

\medskip

%
%\begin{theorem}
%Every Ore graph $G=(V,E)$ has a contagious set of size two.
%\end{theorem}
%
\begin{proof}
	For Dirac graphs that do not coincide with~$DC_n$ any three nodes form a contagious set, but Example~\ref{example_ore_graphs} in Sect.~\ref{sec:ore} demonstrates that this statement is not valid for Ore graphs.
	However, activating three arbitrarily selected nodes with degree~$\frac{n}{2}$ each will infect at least half of the nodes, as we show in Lemma~\ref{lemma_Ore_graphs_spread_of_three_arb_nodes}.
	Such an active set of size three can be obtained by activating two nodes only:
	according to Lemma~\ref{lemma_Ore_suitable_neighbor} there are two nodes~$u,v$ with degree at least~$\frac{n}{2}$, such that both are adjacent to a third node~$w$ of degree at least~$\frac{n}{2}$ as well. Then activating~$u$ and~$v$ will infect~$w$ and subsequently at least half of the nodes.
	
	Thereafter, the infection will reach all nodes unless the graph is isomorphic to $DC_{n}$. This is proven in Lemma~\ref{lemma_Ore_graph_small_contagious_set_almost_there}.
	On the other hand, if the graph is isomorphic to~$DC_n$, then we have already argued in the proof of Theorem~\ref{theorem_generalization_mindegree_arbitrary_thresholds} that $m(G,2) \leq 2$.
	\begin{lemma}
		\label{lemma_Ore_suitable_neighbor}
		In an Ore graph there exists a vertex~$w$ of degree at least $\frac{n}{2}$ that is adjacent to at least two vertices $u,v$ with $\mathrm{deg}(u),\mathrm{deg}(v) \geq n/2$.
	\end{lemma}
	%
	%Due to space constraints the proof was moved to the appendix.
	%
	\begin{proof}
		Let~$S$ be the set of vertices with degree at least $\left\lceil\frac{n}{2}\right\rceil$. We want to show that there exists a vertex in $S$ with two neighbors in $S$. 
		
		First we show that~$S$ must have size at least $\left\lfloor\frac{n}{2}\right\rfloor$: if there is a vertex $x \notin S$, then $x$ has at most $\left\lceil\frac{n}{2}\right\rceil -1$ neighbors, denoted by~$N(x)$. All vertices that do not belong to~$x \cup N(x)$ must belong to $S$ (in order to satisfy the degree constraint for non-adjacent nodes); note that there are at least $\left\lfloor\frac{n}{2}\right\rfloor$ such nodes outside~$\{x\} \cup N(x)$.
		
		If there is no vertex in $S$ with two neighbors in~$S$, then $E(S,\overline{S})\geq (\left\lceil\frac{n}{2}\right\rceil -1)\cdot \left\lfloor\frac{n}{2}\right\rfloor$ as $|S| \geq \left\lfloor\frac{n}{2}\right\rfloor$ and every vertex in $S$ has at least $\left\lceil\frac{n}{2}\right\rceil -1$ neighbors outside $S$.
		Observe that $$\sum_{v\in \overline{S}}\mathrm{deg}(v)\leq |\overline{S}| \cdot \left(\left\lceil\frac{n}{2}\right\rceil -1\right).$$
		Thus, $|E(\overline{S})|$ is bounded above by the difference of this product on the RHS and the lower bound on~$|E(S,\overline{S})|$:
		\begin{equation}
			\label{Eq_lemma_ore_upper_bound}
			|\overline{S}| \cdot \left(\left\lceil\frac{n}{2}\right\rceil -1\right) -  \left(\left\lceil\frac{n}{2}\right\rceil -1\right)\cdot \left\lfloor\frac{n}{2}\right\rfloor = \left(\left\lceil\frac{n}{2}\right\rceil -1\right) \cdot \left(|\overline{S}| - \left\lfloor\frac{n}{2}\right\rfloor\right).
		\end{equation}
		Recall that we showed~$|\overline{S}|\leq \left\lceil\frac{n}{2}\right\rceil$.
		Thus, the bound given in Eq.~(\ref{Eq_lemma_ore_upper_bound}) is nonnegative only if $|\overline{S}| \in \left\{ \left\lfloor\frac{n}{2}\right\rfloor, \left\lceil\frac{n}{2}\right\rceil \right\}$, 
		%
		%%TODO Statement is not true, Dani is looking into this
		% % % MP: Dani has fixed that one
		%
		and hence the upper bound equals $\left\lceil\frac{n}{2}\right\rceil-1$ or~$0$. But~$\overline{S}$ has to be a clique by choice of~$S$ and the degree requirement of Ore graphs. Therefore, the number of edges inside~$\overline{S}$ must be~$\binom{|\overline{S}|}{2}=\binom{\left\lceil\frac{n}{2}\right\rceil}{2}$, which contradicts the upper bound of $\left\lceil\frac{n}{2}\right\rceil-1$ or $0$ on the number of edges inside~$\overline{S}$ if~$n>4$ holds. 
		%
		%%TODO Statement wrong for n=4.
		%
		For~$n \in \{3,4\}$ we recall that every Ore graph has a Hamiltonian cycle~\cite{ore}; observe that the statement of the lemma follows immediately in this case.
		%
		%%TODO add citation
		%
		%\qed
	\end{proof}
	Thus, once we activate~$u,v$, the node~$w$ will become infected and then eventually half of the nodes.
	\begin{lemma}
		\label{lemma_Ore_graphs_spread_of_three_arb_nodes}
		The activation of three vertices with degrees at least $\frac{n}{2}$ each will infect at least half of the nodes in an Ore graph.
	\end{lemma}
	\begin{proof}
		Let $A_0$ consist of three vertices of degree at least $\frac{n}{2}$. Let $A:=\langle A_0 \rangle$, i.e.\ the set of nodes that will eventually be active if we activate~$A_0$. Observe that~$A_0 \subseteq A$ holds by definition.
		
		Assume for the sake of contradiction that~$|A|<\frac{n}{2}$ and recall that~$\overline{A} := V \setminus A$ is the set of nodes that do not become active.
		We claim that each of the vertices in $A \setminus A_0$ must have at least one neighbor in~$\overline{A}$: vertices in~$\overline{A}$ have at most one neighbor in $A$ and thus degree at most~$|\overline{A}|$ each. If $a \in A$ and $b \in \overline{A}$ are non-adjacent, we have that $\mathrm{deg}(a)+\mathrm{deg}(b) \leq |\overline{A}|+|A|-1+|N(a)\cap \overline{A}|$.
		As this quantity has to be at least~$n$ in an Ore graph and~$|A| + |\overline{A}| = n$ holds, $N(a)\cap \overline{A}$ must be non-empty.
		
		Each of the nodes in~$A_0$ has degree at least~$\frac{n}{2}$ by assumption of the lemma, and hence each of them has at least $(\frac{n}{2}-(|A|-1))$ neighbors in $\overline{A}$. Recall that the other~$|A| - 3$ vertices in~$A$ must have at least one neighbor in $\overline{A}$ each. But since each node in~$\overline{A}$ can have at most one neighbor in $A$, otherwise it would be infected, we get
		\begin{align*}
			|\overline{A}| &\geq 3 \cdot \left(\frac{n}{2}-(|A|-1)\right)+(|A|-3)\\
			& =3 \cdot \frac{n}{2}-3\cdot |A|+3+|A|-3\\
			& =\frac{n}{2}+n-2\cdot |A|.
		\end{align*}
		Thus, we have
		$|\overline{A}|+2 \cdot |A| =n+|A|\geq n+\frac{n}{2}$ and the desired contradiction $|A|\geq\frac{n}{2}$ follows.
		%
		%\qed
	\end{proof}
	Next, we show
	\begin{lemma}
		\label{lemma_Ore_graph_small_contagious_set_almost_there}
		Consider an $n$-vertex Ore graph that is not equal to~$DC_{n}$. Then any set of three vertices with degree at least~$\frac{n}{2}$ each is a contagious set.
	\end{lemma}
	\begin{proof}
		Pick any three vertices with degree at least~$\frac{n}{2}$ as seed and let~$A$ be the set of eventually infected vertices. By Lemma~\ref{lemma_Ore_graphs_spread_of_three_arb_nodes} we have $|\overline{A}|\leq |A|$.
		Again every~$b \in \overline{A}$ is adjacent to at most one node in~$A$, otherwise~$b$ would be infected, and hence we have~$\mathrm{deg}(b)\leq |\overline{A}|$.
		Then every node in~$A$ that is non-adjacent to some $b\in \overline{A}$ must have degree at least~$n - |\overline{A}|$ to meet the degree requirement of Ore graphs.
		
		It follows that every vertex $a \in A$ must have at least one neighbor in $\overline{A}$: if~$a$ is adjacent to all vertices in~$\overline{A}$, the claim holds. If~$a$ is non-adjacent to at least one, then we have already shown that~$a$ has degree at least~$n - |\overline{A}| = |A|$. Since node~$a$ can have only~$|A|-1$ neighbors in~$A$, $a$ must have at least one within~$\overline{A}$.

		No vertex in $A$ can have more than one neighbor in $\overline{A}$, since this would imply the existence of a vertex in~$\overline{A}$ with two active neighbors, as $|\overline{A}|\leq|A|$; but this would contradict the choice of~$\overline{A}$.
		Thus, each vertex in $A$ has exactly one neighbor in $\overline{A}$ and we have that $|A| {=} |\overline{A}| {=} n/2$. Notice that~$A$ and~$\overline{A}$ must both be cliques as otherwise two non-adjacent vertices in~$A$ (resp., $\overline{A}$) would have degree less than~$\frac{n}{2}$ each and thus their degrees add up to less than~$n$ contradicting the property of an Ore graph.
		But then the graph is isomorphic to~$DC_{n}$.
		%
		%\qed
	\end{proof}
	\noindent{}This concludes the proof of Theorem~\ref{thm:Ore}.
	%
	%\qed
\end{proof}

\end{document}